\documentclass[]{article}
\usepackage{amsfonts}
\usepackage{amssymb}
\usepackage{graphics}
\usepackage{graphicx}
\usepackage{makeidx}
\usepackage{latexsym}
\usepackage[dvips]{color}
\usepackage[centertags]{amsmath}
\usepackage{amsthm}
\usepackage{newlfont}
\usepackage{stmaryrd}
\usepackage{mathrsfs}
\usepackage{enumerate}
\usepackage[english]{babel}

\newcommand{\be}{\begin{eqnarray}}
\newcommand{\ee}{\end{eqnarray}}

\newtheorem{theorem}{Theorem}[section]
\newtheorem{lemma}[theorem]{Lemma}
\newtheorem{proposition}[theorem]{Proposition}
\newtheorem{corollary}[theorem]{Corollary}

\begin{document}

\title{{\bf A low temperature analysis\\ of the boundary driven Kawasaki Process}}

\author{Christian Maes and  Winny O'Kelly de Galway\\
Instituut voor Theoretische Fysica, KU Leuven}

\maketitle
\begin{abstract}
Low temperature analysis of nonequilibrium systems requires finding the states with the longest lifetime and that are most accessible
from other states.  We determine these {\it dominant} states for a one-dimensional diffusive lattice gas subject to exclusion and with nearest neighbor interaction.  They do not correspond to lowest energy configurations even though the particle current tends to zero as the temperature reaches zero.  That is because the dynamical activity that sets the effective time scale, also goes to zero with temperature.
The result is a non-trivial asymptotic phase diagram, which crucially depends on the interaction coupling and the relative chemical potentials of the reservoirs.
\end{abstract}

\section{Introduction}
The characterization of a macroscopic system of fixed volume and in
thermodynamic equilibrium with a unique heat bath at a given temperature
and chemical potential proceeds from the study of its (grand-canonical)
free energy functional.
At low temperatures energy considerations dominate and the phase diagram
starts from identifying the ground states upon which small thermal
excitations are built and entropic considerations enter. For equilibrium
circumstances then, following the important work in equilibrium statistical mechanics
around 1960-1990, a systematic low temperature analysis has
evolved into a constructive tool, establishing phase transitions and
enabling characterizations of low temperature phases; see 
\cite{minl,dobr,pir,bric3,sin,dom} for some few pioneering examples in the mathematical physics literature.\\

In contrast, low temperature analysis for nonequilibrium systems is
virtually non-existent, at least from a global perspective.  Much
has of course to do with the lack of general principles and with the great
mathematical difficulties in treating spatially extensive processes
under steady nonequilibrium driving.
Recent years have however seen various exactly solvable nonequilibrium
processes very much including some driven diffusive lattice gases \cite{shu,ex,der1,der2,smit}, and various ideas have been launched on the
relevant large deviation theory for nonequilibria.  In particular, a low
temperature analysis for stochastic processes is mathematically very close to what is done in
Freidlin-Wentzel theory for random perturbations of deterministic dynamics.  One must {\it simply} add the nonequilibrium physics and the
relevant examples.  That was part  of the recent paper
\cite{heatb}, where a  scheme was put forward to characterize the
low temperature asymptotics of continuous time jump processes under the
condition of local detailed balance.  The present paper starts from that
same framework to characterize the low temperature stationary condition
of a one-dimensional boundary driven Kawasaki
dynamics.  It is the natural {\it finite temperature} analogue and
extension of the boundary driven symmetric exclusion process.  Particle
reservoirs at the edges of a (large) lattice interval send particles to and receive
particles from the system.  In the bulk, particles are conserved
and hop to nearest neighbor sites following a heat bath dynamics. 
Because the particle reservoirs work at different chemical potentials, a
particle current can be maintained though the system.  Very little is
known about the stationary distribution of the particle configurations and of course the usual Gibbs formalism no longer applies. The low temperature Kawasaki dynamics has been investigated for various reasons, e.g. recently in two dimensions in \cite{land} for tunneling behavior, or for metastability \cite{luck,hol},  for nucleation \cite{bov}, in studies of the spectral gap \cite{mart} etc. but all mostly at detailed balance, \cite{kawa}.\\
In the present paper we break detailed balance.  We start by proving that, asymptotically for very low
temperatures and for positive versus negative chemical potentials at the
edges, the dominant configurations are those that segregate
particles and holes when there is even a small attractive potential between the particles.  Both the current and the dynamical activity go to
zero exponentially fast in the inverse temperature.  That stands in contrast with the case for zero coupling (pure
exclusion dynamics) where the stationary distribution remains concentrated on all possible configurations and a current does of course flow.  We discuss the dominant low temperature attractors and analogous results are exposed also for other parameter values.\\

  The model and the main results will be presented more precisely in the
next section. The discussion of the results is continued in Section \ref{disc}. In Section \ref{kir} we explain what we need from
\cite{heatb}, in particular the set-up of the low temperature asymptotics.  Next, in Section \ref{prof} is contained the detailed
proofs of all results. One should realize here also that mathematical analysis is helpful especially as convincing numerical simulations
become very difficult for larger sizes of the system at very low temperatures. We end in Section \ref{sse} with the proof for the boundary driven exclusion process, that there all configurations are dominant in contrast with the case for weak interaction.

\section{The Kawasaki model and main result}
Consider the lattice interval $I_N = \{1,2,\ldots,N\}$ with each site $i,j,\ldots \in I_N$ either occupied by one particle or left vacant. The particles are treated as indistinguishable so that the configuration space is $\mathcal{G}_N=\{0,1\}^N$. Configurations are denoted by $x,y,z,\ldots \in \mathcal{G}_N$ and $x(i) \in \{0,1\}$ stands for the number of particles at site $i$. \newline 
We take a nearest neighbor interaction between the particles of the form
\be\label{ene}
E(x)= -\kappa \sum_{i=1}^{N} x(i)x(i+1),
\ee
When the coupling $\kappa>0$ the particles attract each other, $\kappa >0$ makes the interaction repulsive and $\kappa=0$ will correspond to the simple exclusion process. Note that for $\kappa\neq 0$ the particle--hole symmetry is (in general) broken.\newline
The dynamics is composed of two parts, nearest neighbor hopping of particles in the bulk and creation or annihilation at the boundaries of $I_N$. We denote by $x^{i,j}$ the configuration obtained from $x$ by interchanging the occupation at $i$ and $j$:
\[
x^{i,j}(k) = %
\begin{cases}
x(k) & \text{if } k\neq i, k\neq j; \\
x(i) & \text{if } k=j; \\
x(j) & \text{if } k=i
\end{cases}
\]
The only allowed such exchanges are between nearest neighbors $j=i\pm 1$.  Their rate  is taken as
\begin{equation}\label{dif}
k(x\rightarrow x^{i,j})=\exp\left[-\frac{\beta}{2}\left(E(x^{i,j})-E(x)\right)\right],\quad \vert i-j\vert = 1
\end{equation}
Note that this particular choice of rates is rather arbitrary up to the natural (local detailed) condition that
\[
\log \frac{k(x\rightarrow x^{i,j})}{k(x^{i,j}\rightarrow x)} = \beta\, \left(E(x)-E(x^{i,j})\right)
\]
is the entropy flux towards the environment due to the bulk occupation exchange  $x\rightarrow x^{i,j}$.\newline
For the boundary sites $i=1,N$ we denote by $x^i$  the configuration obtained from $x$ by flipping the occupation:
\[
x^i(k)  = %
\begin{cases}
1-x(i) & \text{if } k=i; \\
x(k) & \text{if } k\neq i
\end{cases}
\]
The rates of birth and death of particles at $i=1,N$ is then written as 
\begin{equation}\label{rea}
k(x\rightarrow x^i)= e^{\frac{\beta \mu_i}{2} (1-2x(i))}\exp\left[-\frac{\beta}{2}\left(E(x^{i})-E(x)\right)\right]
\end{equation}
so that the ratio
\begin{equation}\label{tdi}
\log \frac{k(x\rightarrow x^i)}{k(x^i\rightarrow x)} = \beta \mu_i \big({\cal N}(x^i) -{\cal N}(x)\big) + \beta \big(E(x) -E(x^i)\big), \quad i=1,N
\end{equation}
equals the entropy flux to the left ($i=1$) or right ($i=N$) particle reservoir imagined with chemical potential $\mu_1$ respectively $\mu_N$, and particle number ${\cal N}(x) := \sum_j x(j)$; in particular, ${\cal N}(x) -{\cal N}(x^i)= 2x(i)-1$. We repeat however that also here other choices than \eqref{rea} give that same thermodynamic interpretation but they would present another kinetics which, for nonequilibrium, does matter. For instance, the rate for annihilation could be fixed at one, independent of temperature, which would change the time scale at which the transition happens compared with \eqref{rea}, but by suitable changes in the creation rates,  that would remain fully compatible with \eqref{tdi} and its thermodynamic interpretation. Much more than in equilibrium therefore we expect non-universal behavior also at the critical zero-temperature.\\

The above dynamics defines an irreducible Markov process $X_t$ on $\mathcal{G}_N$ with unique stationary distribution $\rho = \rho_{N,\beta,\mu_1,\mu_N,\kappa}$.  It is the boundary driven Kawasaki dynamics that is the main subject of this paper.  For $\kappa=0$ the model is known as the boundary driven simple exclusion process for which the matrix product representation gives full control of the $\beta\uparrow +\infty$ limit of the stationary regime, \cite{der1,der2}.  We will use it in Section \ref{sse}. Another solvable case occurs when $\mu_1=\mu=\mu_N$, for equal chemical potentials.  Then, the stationary regime is in fact an equilibrium regime with stationary
distribution given by the grand-canonical Gibbs distribution
\[
\rho^{\text eq}(x) = \rho_{N,\beta,\mu_1=\mu_N=\mu,\kappa}(x) = \frac 1{\cal Z}\exp \big(\beta\mu{\cal N}(x) - \beta E(x)\big)
\]
It is easy to check that $\rho^{\text eq}$ is a reversible distribution for the dynamics \eqref{dif}--\eqref{rea} when $\mu_1=\mu=\mu_N$, as expressed in the (global) detailed balance relation
\[
\frac{\rho^{\text eq}(x)}{\rho^{\text eq}(y)} = \exp[\mu\beta\big({\cal N}(x)-{\cal N}(y)\big) -\beta\big(E(x)-E(y)\big)]= \frac{k(y\rightarrow x)}{k(x\rightarrow y)}
\]

\subsection{Main result}
For the boundary driven Kawasaki process defined above we investigate the large $\beta$ behavior of $\rho_{N,\beta,\mu_1,\mu_N,\kappa}$ for various choices of the other parameters.  We assume physically that the low temperature variation  of the chemical potentials of the reservoirs is zero; so that we can keep $\mu_{1,N}\equiv\mu_{L,R}$ constant. Note that these are multiplied with $\beta$ in \eqref{rea} so that we effectively get to deal with either births or deaths at the edges.  Similarly, the interaction coupling $\kappa$ is also thought to be temperature independent.\\
  A first interesting case concerns an attractive potential ($\kappa > 0$) when the left and right chemical potentials have a different sign, 
say $\mu_1\equiv \mu_L >\kappa >0 >\mu_R \equiv\mu_N$, $\vert\mu_R\vert >\kappa$. It is tempting to think that in the $\beta\uparrow +\infty$-limit, the distribution settles to be uniform over the ground states of the equilibrium lattice gas with energy \eqref{ene} and with boundary conditions $x(1)=1,x(N)=0$.  Independent of the question why energy alone would be decisive for nonequilibrium stationary distributions, that is in fact not entirely correct (and entirely wrong for $\kappa=0$). The dominant low temperature configurations are of the form $x = \eta_{p,q} := (1,1,1,\ldots,1,0,0\ldots,0)$ with $p\geq 3$ occupied sites followed by $q\geq 2$ vacant sites.\\

To give a precise sense to the low temperature asymptotics, we 
introduce the notation
$f(\beta) \asymp e^{\beta h}$ for $\lim_{\beta\rightarrow\infty}\frac{1}{\beta}\log f(\beta) =
 h$.  The states $x$ with $\rho(x)\asymp 1$ are called {\it dominant}.

\begin{theorem}\label{thm:thm1}
For $N\geq 5$ and with $\mu_L >\kappa >0 > \mu_R, |\mu_R|> \kappa$, 
\begin{equation}
\rho(x)\asymp 1 \mbox{ iff } x= \eta_{{p,q}}
\end{equation}
for some $p\geq 3, q\geq 2, p+q=N$.   For all other $x\in {\mathcal G}_N, \rho(x) \asymp e^{-\beta \alpha}$, for some $\alpha>0$.
\end{theorem}
In the same case but for $N<5$, there appears a unique dominant state: $11$ for $N=2$, $110$ for $N=3$ and two dominant states $1100$ and $1110$ for $N=4$.\\
We also give the results for the other parameter regimes without proofs. They will be summarized in the phase diagram of the next section.
\begin{itemize}
\item Region \MakeUppercase{\romannumeral 1} consists of the patches $[\kappa >0,  \mu_L> 0, \mu_R>0]$, $[0<-\kappa < \mu_L <\mu_R]$, $[\kappa>-\mu_R>0,\mu_L>0]$ and $[0<-\kappa < \mu_L <\mu_R]$. The fully occupied state $(1,1,\ldots,1,1)$ is the unique dominant state.
\item Region \MakeUppercase{\romannumeral 2} consists of two subsection, \MakeUppercase{\romannumeral 2}$_A$ $[0<\mu_L<-\kappa<\mu_R]$) and \MakeUppercase{\romannumeral 2}$_B$ $[0<\mu_R<-\kappa<\mu_L]$. The set of dominant states depends on whether $N$ is odd or even. 
\begin{itemize}
\item For odd $N$, the unique dominant state is $(1,0,1,0,\ldots,1,0,1)$ for both patches. 
\item For even $N$, the dominant states that are shared by both \MakeUppercase{\romannumeral 2}$_A$ and \MakeUppercase{\romannumeral 2}$_B$ are those which have either an extra vacancy or an extra occupied site compared with the odd case (except for two configurations, discussed below). Those which are shared for $A$ and $B$ are of the form $(1,0,1,0,\hdots,1,0,0,1,\hdots,0,1)$ or $(1,0,1,\hdots,0,1,1,0,1,\hdots,0,1)$. \\ The difference between \MakeUppercase{\romannumeral 2}$_A$ and \MakeUppercase{\romannumeral 2}$_B$ is that $(1,0,1,0,\hdots,1,0,1,1)$ is dominant for $A$ but not for $B$ since its preferred successor is $(1,0,1,0,\hdots,1,1,0,1)$ with rate (2) which yields the maximal life-time (see section \ref{kir}) when the parameters lie in $A$, and $(1,0,1,0,1\hdots,0,1,0)$ (with rate (2)) when the parameters lie in $B$ (here, the life-time is less than the life-time of dominant states in $B$). Similarly, one can argue that the state $(1,1,0,1,0\hdots,1,0,1)$ is not dominant for patch $A$ whereas it is for $B$.
\end{itemize}
\item Region \MakeUppercase{\romannumeral 3} consists of the patches $[0<\mu_L<\mu_R<-\kappa]$ and $[0<\mu_R<\mu_L<-\kappa]$. For odd $N$ the unique dominant state is $(1,0,1,0,\ldots,1,0,1)$. For $N$ even the dominant states are all states that have an extra vacancy compared with the odd case and where the occupied sites are not neighboring, i.e., of the form; $(1,0,1,\hdots, 1,0,0,1,\hdots,0,1)$;
\item Region \MakeUppercase{\romannumeral 4} consists of the lower left half where $[\mu_R,\kappa<0,\mu_L>0]$. Here, the dominant states are all states where the occupied sites are not neighboring, i.e., of the form;\\ $(1,0,0,\ldots,0,1,0,0,\ldots,0,1,0,0,\ldots,0,\ldots, 0,0,\ldots,0,0)$,
\item Region \MakeUppercase{\romannumeral 5}, where $[-\mu_R>\kappa>0,\mu_L>0]$, the dominant states are given by Theorem \ref{thm:thm1};
\end{itemize}
Although we give only a proof for region \MakeUppercase{\romannumeral 5}, determining the dominant states for the other regions is less difficult since they appear to be rater trivial and expected. Furthermore, region \MakeUppercase{\romannumeral 5} contains a wide variety of possible dominant states which makes it by far the most interesting patch in the phase diagram.

\section{Discussion}\label{disc}

The dynamics of low temperature lattice gases in in general dominated by domain wall movements; see e.g. \cite{priv} and in particular the contribution by S.J.~Cornell and more recently in \cite{goi} for zero-temperature Kawasaki dynamics in two dimensions, or more generally in \cite{ex} also for (driven) nonequilibrium models.  That remains true for the boundary driven case of the present paper but there is an additional element of transport.  For the situation of Theorem \ref{thm:thm1} particles are created to the left and they disappear towards the right.  The dominant states are states with one ``interface'' and low temperature motion can be pictured as a random walk of that interface on a time scale which is exponentially long in $\beta$.
The particle current then naturally also appears to go to zero with low temperature, and is exponentially small.  One could think that the system becomes more and more equilibrium-like as the current gets smaller, but that is not the case.  The reason is that the dynamical activity, which is basically the rate of escape from the dominant states and which sets the time-scale, also goes to zero exponentially fast at the same rate.  The total result is a nonequilibrium behavior, with dominant states that do not correspond to minima of the energy.  Detailed aspects of low temperature current and dynamical activity will appear in another paper, jointly with Karel Neto\v{c}n\'{y}.\\

To visualize the full zero-temperature phase diagram we fix
$\mu_L>0$, and we indicate the different regions as a function of the interaction strength $\kappa$  and the right chemical potential $\mu_R$. The roman numbers indicate patches of the diagram for which the parameters yield the same set of dominant states:\\

\setlength{\unitlength}{1cm}
\begin{center}
\begin{picture}(6,6)(-3,-3)
\put(-3,0){\vector(1,0){6}}
{\linethickness{0.4mm}
\put(0,0){\line(-1,0){3}}}
\put(3,-0.2){$\kappa$}
{\linethickness{0.4mm}
\put(0,-3){\line(0,1){3}}}
\put(0,-3){\vector(0,1){6}}
\put(0.2,3){$\mu_R$}
\put(0.4,0.2){{\Huge \MakeUppercase{\romannumeral 1} }} 
\put(-2.3,2.3){{\LARGE \MakeUppercase{\romannumeral 2}$_B$ }}  
\put(-1.3	,0.3){{\LARGE \MakeUppercase{\romannumeral 2}$_A$ }}
\put(-2.8,1.2){{\Huge \MakeUppercase{\romannumeral 3} }}
\put(-2,-1.5){{\Huge \MakeUppercase{\romannumeral 4} }}	
\put(0.4,-2.5){{\Huge \MakeUppercase{\romannumeral 5} }}
{  \thicklines \put(-3,3){\line(1,-1){6}}}
{\linethickness{0.4mm}
\multiput(-3,1.5)(0.4,0){15}
{\color{red}\line(1,0){0.2}}}
{\linethickness{0.4mm}
\multiput(-1.5,0)(0,0.18){15}
{\line(0,1){0.4}}}
\put(3,1.3){$\mu_L$}
\put(-1.8,-.3){$-\mu_L$}
\end{picture}\\

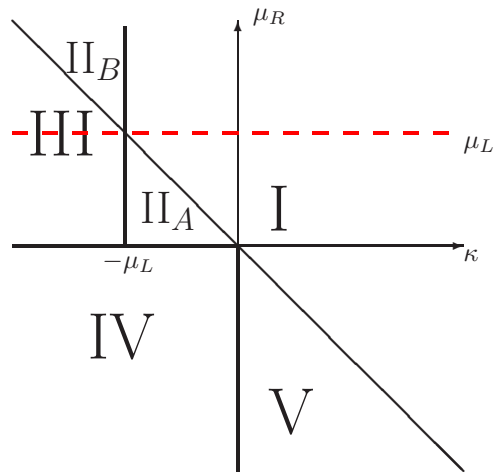
FIG. 1. The dashed line shows the equilibrium condition (detailed balance) where $\mu_R=\mu_L$. The line $\kappa=0$ is treated in Section \ref{sse}; there all bulk configurations remain supported.
\end{center}

Although Theorem \ref{thm:thm1} only applies to $[\mu_L>-\mu_R>\kappa>0]\subset[-\mu_R>\kappa>0,\mu_L>0]$, the statements used to prove Theorem \ref{thm:thm1} are quite similar as to proving it for its complement in $[-\mu_R>\kappa>0,\mu_L>0]$.

\section{Low temperature asymptotics}\label{kir}
 The present section starts from general definitions and assumptions that are all verified in the case of the driven Kawasaki dynamics of the previous section.  Our notation will however refer more generally to an irreducible continuous time Markov jump process on a finite state space $K$ with transition rates $k(x,y;\beta)$ for $x\rightarrow y$ that depend on a real parameter $\beta$ (to be interpreted as inverse temperature as in \eqref{dif} and \eqref{rea} for $K={\mathcal G}_N$).  We also assume that $k(x,y,\beta)>0$ iff  $k(y,x,\beta)>0$.  The unique stationary distribution is denoted by $\rho = \rho_\beta$.\\
   We follow the set-up of \cite{heatb} in assuming the existence of the (logarithmic) limit
 \[
 \phi(x,y):= \lim_{\beta\rightarrow\infty}\frac{1}{\beta}\log k(x,y,\beta), \quad x,y\in K
 \]
Thus, $k(x,y,\beta) \asymp \exp[\beta \phi(x,y)]$.
Then, the escape rates have the asymptotics $\xi(x) :=\sum_y k(x,y)\asymp e^{-\beta \Gamma(x)}$ with $\Gamma(x)=-\max_y \phi(x,y)$. The asymptotic life-time of a state $x$ is thus $e^{\beta \Gamma(x)}$ and  when the system makes a jump from $x$, the probability to jump to state $y$ asymptotically goes like $p(x,y):=k(x,y)/\xi(x)\asymp e^{-\beta U(x,y)}, x\neq y$ where $U(x,y) := -\phi(x,y) - \Gamma (x)\geq 0$. We put $U(x,x) = +\infty$.  For all $x$ there is at least one state $y\neq x$ for which $U(x,y)=0$ --- we call these states \textit{preferred successors} of $x$.\\

 A useful low temperature representation of the stationary distribution is in terms of the Kirchoff formula \cite{heatb}.  We make the state space $K$ into a graph with its elements $x$ as vertices and edges $x\sim y$ for these pairs where $k(x,y;\beta) > 0$ (iff $k(y,x;\beta) >0$) assuming that this does not depend on $\beta >0$.  We denote by $\mathcal{T}_x$ the in-tree to $x$ defined for any tree $\mathcal{T}$ on $K$ by orienting every edge in $\mathcal{T}$ towards $x$.\\
 In \cite{heatb} a Kirchoff formula for the low temperature stationary distribution was obtained:

\begin{proposition}
 For all $x\in K$,
\begin{equation}
\rho_\beta(x) \asymp \exp{-\beta[\Psi(x) - \max_{y\in K}\Psi(y)]}
\end{equation}
where 
\begin{equation} \label{thet}
\Psi(x):=\Gamma(x)-\Theta(x)
\end{equation}
for $\Theta(x):=\min_{\mathcal{T}}U(\mathcal{T}_x)$ and $U(\mathcal{T}_x):=\sum_{(y,z)\in {\mathcal T}_x} U(y,z)$.
\end{proposition}
Naturally then, we call a state $x$ {\it dominant} if ($\rho_\beta(x) \asymp 1$ or) $\Psi(x) \geq \Psi(y)$ for all $y\in K$.  If a state $y$ is not dominant, $\rho(y)\asymp e^{-\beta \alpha}, \alpha>0$.  Note that two considerations combine; the first term in $\Psi$ being a measure of the life-time $\Gamma(x)$ and the second term, $U(\mathcal{T}_x)$ relating to the accessibility from other states.   Under equilibrium conditions, these reduce to energy considerations only.  Indeed, suppose say $k(x,y;\beta) = \exp[\beta E(x) - \beta \Delta(x,y)]$. Then $\Gamma(x) = -E(x) + v(x)$ and $U(y,z) = \Delta(y,z) - v(y)$ for $v(z):= \min_u \Delta(z,u)$.  Assuming detailed balance $\Delta(x,y)=\Delta(y,x)$ we get $\Theta(x) = v(x) + C$ so that then $\Psi(x) = -E(x)$ up to a constant $C$.\\

A path $D=(x_1,\ldots,x_n)$ is an ordered sequence of oriented edges on the graph $(K,\sim)$ for which we denote $U(D):=U(x_0,x_1)+\hdots + U(x_{n-1},x_n)$, growing with the number of edges not following a preferred successor. Fixing the beginning $x_0=x$ and end $x_n=y$ we write 
\[
{\cal U}(x,y) = \min_D U(D)
\]
over all paths from $x$ to $y$. Any path realizing that minimum is called a `preferred path'.\\
A non-empty set $A\subset K$ is called an \textit{attractor} when
\begin{enumerate}
\item ${\cal U}(x,y)=0$ for all $x\neq y \in A$;
\item ${\cal U}(x,y)>0$ for all $x\in A$ and $y\in K \setminus A$.
\end{enumerate}
As an example, it is easy to see for the case of Theorem 1 that the pair $\{(1,1,1,0,0),(1,1,0,1,0)\}$ is an attractor for ${\mathcal G}_{N=5}$ with rates \eqref{rea} and \eqref{dif}.

\begin{proposition}
\label{bep:bep2}
Let ${\cal A}=\cup_iA_i\subset K$ collect all states that are elements of an attractor $A_i$.  Then, for all $y\in K,\; {\cal U}(y,x)=0$ for some $x\in {\cal A}$.
\end{proposition}

\begin{proof}
Fix a vertex $y\in K$ and consider the oriented subgraph $\mathcal{G}_y \subset \mathcal{G}$ obtained by considering all vertices and oriented edges in the set of paths $D=(x_1=y,x_2,x_3\ldots,x_n)$ that start in $y$ and go along consecutive preferred states $(x_i,x_{i+1})$ with $U(x_i,x_{i+1})=0$.  We claim that some vertex of the graph $\mathcal{G}_y$ is contained in ${\cal A}$.  The point is simply that some attractor must be contained in $\mathcal{G}_y$.  Since $\mathcal{G}_y$ is a general oriented connected graph in which each vertex $v$ has at least one outgoing edge $(v,w)$, it means quite generally that an arbitrary (but finite) oriented connected graph in which each vertex has an outgoing edge (such as $\mathcal{G}_y$) always contains an subgraph ${\cal C}$ in which all vertices in $\cal C$ can be reached from any other vertex in $\cal C$ and where all edges touching ${\cal C}$ are incoming (and not outgoing).  That can however easily be shown by induction.\\  Suppose indeed such a given oriented graph ${\cal B}$ with the (attractor) subgraph $\cal C$.  Let us fix the vertex set of ${\cal B}$ but add one oriented edge to it.  We only need to consider the case where that edge $(v,w)$ is outgoing from ${\cal C}$, i.e., $v\in {\cal C}, w \notin {\cal C}$.  Let us now add to ${\cal C}$ all edges $(w,w_2)$ in ${\cal B}$ and consider ${\cal C}' = {\cal C} \cup \{(w,w_2)\}$.  If there is no outgoing edge from ${\cal C}'$ we are finished.  If not, it must be edges, of the form $(w_2,w_3)$, which we again add etcetera.  Since the graph is finite, there is moment where no new vertices $w_i$ appear in the construction and that so obtained maximal set makes an attractor. Adding a vertex $w$ to ${\cal B}$ with just one edge connecting to ${\cal B}$ should only be considered when that edge is  of the form $(v,w)$ with $v\in {\cal C}$.  But then, the set ${\cal C}\cup(v,w)$ makes an attractor. 
\end{proof}

The Proposition gives a clear picture of the accessibility of states.  In the set $K$ there are a number of disjoint attractors and every vertex can get to one of these by a path of preferred successors.  To leave an attractor $ A_i$ means to ``pay'' $U(v,w)>0$ over an edge $(v,w)$ with $v\in A_i, w\notin A_i$.  Proving Theorem 1 is therefore first characterizing the attractors in ${\cal G}_N$, and then to find the dominant states by comparing lifetimes.

\section{Proof of Theorem 1}\label{prof}
Let $K={\mathcal G}_N$ with $N\geq 5$. Denote a state $x$ by $x=(p_0,q_0,\hdots, p_n,q_n)$, $\sum_{i=0}^n (p_i + q_i)=N$, where the $p_i$ stand for the number of consecutively occupied sites and the $q_i$ for the number of consecutively vacant sites. A priori we could have $p_0=0$ or $q_n=0$ but to make sense we require $p_i,q_{n-i} >0$ for all $n \geq i\geq 1$. We consider the following subset of states  
\begin{equation}\label{nota}
\mathcal{P} = \left\{(p_0,q_0,\hdots, p_n,q_n) \vert p_i \geq 3 \mbox{ for all } i, q_i \geq 3 \mbox{ for all } i < n, q_n\geq 2  \right\}
\end{equation}
and we let the set $\mathcal{E}$ of states that are obtained by first taking $x\in \mathcal{P}$ and then making just one occupation switch of the form $\hdots 11100\hdots\rightarrow \hdots 11010 \hdots $ or $ \hdots 000111\hdots$ to $\hdots 001011\hdots$.
\begin{lemma} \label{bep:bep3}
Take $x\in \mathcal{P}$. If $U(x,y)=0$, then $y\in \mathcal{E}$ and $U(y,z)=0$ implies $z=x$.
\end{lemma}
\begin{proof}
Take $x\in \mathcal{P}$. Apart from a possible switch, the state $x$ could also change via the annihilation of a particle at the left boundary or the creation of a particle at the right boundary. But the rates of the latter to occur have a factor $e^{\mu_R \beta/2}$ or $e^{-\mu_L \beta/2}$ so that the transition to a switched state is always preferred for states in $\mathcal{P}$. \newline
Let now $y\in \mathcal{E}$  be that preferred successor to $x$:
\[
y \equiv  \left( p_0,q_0,\hdots,p_i,q_i,\left[  11\hdots 010\hdots 0  \right], p_{i+2}, \hdots, 00\right).
\]
It is trivially checked that the preferred successor to $y$ is again $x$.
\end{proof}

\begin{corollary}\label{ax}
Let $x\in \mathcal{P}$. Then, the set containing $x$ and all its preferred successors is an attractor ${\cal A}_x$ that has empty intersection with any other ${\cal A}_y$, similarly made from $y\in \mathcal{P}, y\neq x$. 
\end{corollary}
\begin{proof}
The fact that it is an attractor is immediate. But also, $x\neq y \in \mathcal{P}$ cannot be in the same attractor since that would imply that $\mathcal{U}(x,y)=0$. We know however that all preferred paths from $x$ come back to $x$ in two steps.
\end{proof}

Next comes the opposite, that any attractor corresponds also with exactly one element in $\mathcal{P}$.
\begin{lemma} \label{bep:bep4}
The number of attractors is exactly the cardinality of $\mathcal{P}$.
\end{lemma}
\begin{proof}
It suffices to show that  for any $y\notin \mathcal{P}\cup \mathcal{E}$ there must be a path of consecutive preferred successors from $y$ to some $x\in \mathcal{P}$.\\
Write $y=(p_0,q_0,\hdots, p_n,q_n)$.  Clearly we can go via consecutive preferred successors to a state where $p_0\neq 0$ and $q_0\neq 0$.  In fact, it is easily checked that we can even obtain $p_0\geq 3$ and $q_n\geq 2$ just moving along preferred successors.  For example, from $y=(p_0=2,q_0,\hdots, p_n,q_n)$ with otherwise $p_i\geq 3$ and $ q_i\geq 3$ except possibly for $q_n\geq 2$, there is a path along preferred successors to $(3,q_0-1,\hdots, p_n,q_n)$. So we can as well assume from the start that $p_0\geq 3, q_n\geq 2$. Imagine now as a further possibility that 
  $y=(p_0,q_0,\hdots, p_n,q_n)$ has exactly one $j\neq 0$ with $1\leq p_{j}\leq 2$ (and all others again verifying $p_i\geq 3$ and $ q_i\geq 3$ except possibly for $q_n\geq 2$).   Then, one can construct a path from $(p_0,q_0,p_1,q_1,\hdots,p_{j-1},q_{j-1} ,p_j,q_j,\hdots ,p_n,q_n)$ to $ (p_0,q_0,\hdots p_{j-1}+p_j,  q_{j-1}+q_j,\hdots p_nq_n) \in \mathcal{P}$ along preferred successors.  The same applies of course to the situation when there is one $q_k\leq 2$, then $p_k$ is added to $p_{k+1}$, etc.  Since we can thus treat all cases where there is one $p_i$ or $q_i$ which is not appropriate to belong to $ \mathcal{P}$, we can work with induction on the number of ``bad'' intervals, i.e., those which fail to have $p_i\geq 3$ or $q_i\geq 3$.  One picks then the last bad interval to redo the joining of above, and one ends with one bad interval less.  The induction can therefore proceed.
\end{proof}

We look back at the attractors ${\cal A}_x, x\in {\cal P}$ of Corollary \ref{ax}.  We also consider now as in Proposition 1 for any tree on ${\cal G}_N$ the in-tree $\mathcal{T}_x$ to $x$ by orienting all edges toward $x$; there is then a unique path $D(z,\hdots , x)$ from any vertex $z\neq x$ to $x$ along the tree.  Suppose now $z\in \mathcal{P}$.  Since $z$ spans the attractor ${\cal A}_z$, there exists an edge $(u,v)\in \mathcal{T}_x$ which is pointing out of the attractor, i.e., $u\in A_z$, $v\notin A_z$ with the property 
\begin{equation}\label{bep:bep9}
U(u,v)\geq \kappa
\end{equation}
Furthermore, all states in  the $|\cal P|$ attractors must be connected to $x$.  But we have just seen that to leave an attractor the cost is at least $\kappa$.  Therefore, whenever $x\in \mathcal{G}_N$,
\begin{equation}
\label{bcol:bcol1}
U(\mathcal{T}_x) =\sum_{(u,v)\in\mathcal{T}_x} U(u,v)\geq(|\cal P|-1)\,\kappa
\end{equation}

Let us define the candidate dominant states, as in Theorem 1, $\mathcal{D}=\{x\in \mathcal{G}_N \vert x\equiv \eta_{p,q}, p\geq 3, q\geq 2 \}\subset \mathcal{P}\subset \mathcal{G}_N$
 where $p+q=N$.  In the notation of \eqref{nota}, $\eta_{p,q} = (p,q)$. 
\begin{lemma}
\label{bep:bep5}
$\mathcal{U}(\eta_{p,q},\eta_{p+1,q-1}) = \mathcal{U}(\eta_{p+1,q-1},\eta_{p,q})=\kappa$ for $q\geq 3$. 
Moreover, for
$x = (p_0,q_0,\hdots, p_n,q_n)\in \mathcal{P}\backslash \mathcal{D}$, 
\[
\mathcal{U}((p_0,q_0,\hdots, p_n,q_n),(p_0,q_0,\hdots, p_n-1,q_n + 1))=\kappa
\]
and by iteration
\[
\mathcal{U}((p_0,q_0,\hdots, p_n,q_n),(p_0,q_0,\hdots, p_{n-1},q_{n-1}'))=(p_n-2)\kappa
\]
where $q_{n-1}' = q_{n-1} + p_n + q_n$.
\end{lemma}
\begin{proof} 
Since $\eta_{p+1,q-1}\notin {\cal A}_{\eta_{p,q}}$, it follows from \eqref{bep:bep9} that $\mathcal{U}(\eta_{p,q},\eta_{p+1,q-1})\geq\kappa$. It thus suffices to make a path between the two states with total cost $\kappa$. We make it as follows, from $\eta_{p,q}\equiv (p,q)$ to $\eta_{p+1,q-1}\equiv (p+1,q-1)$, as illustrated for $p=4,q=3$: 
\[
(1111000)\rightarrow(1110100)\rightarrow(1101100)\rightarrow(1011100)\rightarrow(0111100)\rightarrow(1111100)
\]
Since all edges except the second one are along preferred successors the result follows.  The rest of the Lemma follows in the very same way.
\end{proof}
Recall the minimum $\Theta $ over trees as defined following \eqref{thet}.

\begin{proposition} \label{bep:bep7}
For all $x\in \mathcal{D}$, $ \Theta(x)= (\vert \mathcal{P}\vert-1)\,\kappa$.
 \end{proposition}
\begin{proof} We construct a path from any vertex in $\mathcal{G}_N$ to $x=(p,q)$ such that the collection of all edges forms an in-tree $\mathcal{T}_x$ for which the minimum of $U(\mathcal{T}_x)$, $\Theta(x)$, is reached. \newline
To connect the states in $\mathcal{D}\backslash x$ with $x$, we use the construction of Lemma \ref{bep:bep5}.
The vertices in $\mathcal{D}$ can be totally ordered as $\eta_{3,N-3},\hdots , \eta_{N-2,2}$. The path from $y= (3,N-3)$ to $x= (p,q)$ will pass along all ${p',q'}$, $3< p' < p$, and similarly the path from $(N-2,2)$ to $(p,q)$ passes through all the remaining states in $\cal D$.
If we collect all the involved edges $(y,z)$ so far, their sum equals $\sum_{(y,z)} U(y,z) = (\vert \mathcal{D}\vert-1)\,\kappa$.\newline
The next step is to add all paths from vertices in $\mathcal{P}\backslash \mathcal{D}$ to $x=\eta_{p,q}$. Consider the set
\[
\mathcal{Z}_{x}=\{y\in \mathcal{P}\backslash\mathcal{D} \vert y = (p,q_0,p_1,\hdots,p_n,q_n) \}
\]
where $p$ in $y$ is fixed. Obviously, extra variation over $p$ generates all of $\cal D$, or $\cup_{x\in \mathcal{D}} \mathcal{Z}_x = \mathcal{P}$. \newline 
Take any such $y\in \mathcal{Z}_{x}$. From the constructions in Lemma \ref{bep:bep5} it follows that
\[
U\left(D\left((p,q_0,\hdots , p_n,q_n),\hdots,(p,q_0,\hdots , p_{n-1},q_{n-1}')\right)\right)=(p_n-2)\,\kappa
\]
where $q'_{n-1}=p_n+q_n+q_{n-1}$. Iterating this from right to left, we get a path to $x$ with
\[
U\left(D\left((p,q_0,\hdots , p_n,q_n),\hdots,x\right)\right)=\sum_{i=1}^n(p_i-2)\kappa
\]
and note that all vertices on that path belong to $\mathcal{Z}_{x}$.  There may be others left, so take then
 $y'\in \mathcal{Z}_{x}$ with $y'$ not covered by that path. As before, we construct a path from $y'$ to $x$ but we stop at the first state $z$ along that path that was also on the previous path.  This procedure can now be repeated by consecutive choices of other states in $\mathcal{Z}_{x}$, always ensuring that we avoid overlap.  At the end, again a set of edges $(v,w)$ appears in which we have covered all $\mathcal{Z}_{x}$, and the total cost (in terms of sums over $U(v,w)$) equals $\vert \mathcal{Z}_x \vert \kappa$.
Let us now take $z\in \mathcal{Z}_{y}, y\neq x$, with $y\in {\cal D}$.  As above we make an oriented graph with all edges pointing to $y$.  Since $y$ is already connected to $x$, we are done.   As a consequence, the total cost for connecting all of $\mathcal{P}$ to $x\in \mathcal{D}$ is $\sum_{y\in \mathcal{D}}\vert \mathcal{Z}_y\vert \kappa + (|{\cal D}|-1)\,\kappa =(\vert \mathcal{P} \vert -1)\kappa$.\newline
  Finally, there are states that have not been covered but surely they do not belong to ${\cal P}$.  By Proposition \ref{bep:bep2} it follows that for any such state $y$, there exists a vertex $z\in \mathcal{P}$ such that $\mathcal{U}(y_i,z_i)=0$. Since any $z\in \mathcal{P}$ is connected via a path to $x$, we can extend the already defined paths in the graph to a spanning tree. The conclusion then follows from the fact that $(\vert \mathcal{P}\vert-1)\kappa$ is also a lower bound for $U(\mathcal{T}_x)$ as was written in \eqref{bcol:bcol1}.
\end{proof} 

Denote by $\mathfrak{T}_x=\arg\min_{\mathcal{T}}U(\mathcal{T}_x)$ the minimum over trees.
Let $x\in \mathcal{P}$, then all states in the $\vert \mathcal{P}\vert$ attractors are connected to $x$ via a unique path $D$ in some $\mathcal{T}_x$. That means that there is (at least one) edge $(v,w)\in D$ where the path leaves an attractor. Let $\mathcal{B}_{\mathcal{T}_x}=\{(v,w)\in \mathcal{T}_x: \exists y\in \mathcal{P}\backslash x: v\in A_y, w\notin A_y\}$ be the set of all such edges in $\mathcal{T}_x$. From the definition of $A_y$ it follows that any 
of these edges have the property $U(v,w)\geq \kappa$, which represents the cost to leave the attractor along $(v,w)$. Since there are $\vert \mathcal{P}\vert$ attractors, $\vert \mathcal{B}_{\mathcal{T}_x}\vert\geq \vert \mathcal{P}\vert-1$. From Proposition \ref{bep:bep7} it follows that $\mathcal{B}_{\mathfrak{T}_x}=\{(v_1,w_1),\hdots, (v_{\vert \mathcal{P}\vert-1},w_{\vert \mathcal{P}\vert-1}) \}$ for $x\in \mathcal{D}$. This also shows that $\mathcal{B}_{\mathfrak{T}_x}$ is the set of edges such that $U(v,w)\neq 0$ in $\mathfrak{T}_x$ for $x\in \mathcal{D}$ and we claim that there are no other states than those in $\mathcal{D}$ that share this property.

\begin{lemma} \label{bep:bep10}
For all $y\in \mathcal{P}\backslash\mathcal{D}$, there exists an edge $(v,w)\in \mathfrak{T}_y:$ $U(v,w)\geq \kappa$ and $(v,w)\notin \mathcal{B}_{\mathfrak{T}_y}$. 
\end{lemma}									
\begin{proof}
We show that given $y\in \mathcal{P}\backslash \mathcal{D}$ there is $x\in \mathcal{P}$ such that along any path $D(x,\hdots,y)$ an edge $(v,w)$ exists such that $U(v,w)\geq \kappa$ where both $v,w$ are not in an attractor $A_y$. \newline
Let $y=(p,q_0,p_1,q_1)$ where $p$ in $y$ is fixed and $x=(p,q)$. It is easily checked that any path $D$ from $x$ to $y$ contains such edge (whereas the opposite is not true). Then it is certainly true for $y=(p,q_0,\hdots,p_n,q_n)$, $n>1$. Variation over $p$ generates all of $\mathcal{P}\backslash \mathcal{D}$.
\end{proof}

\begin{lemma} \label{bep:bep11}
For all $y\in \mathcal{P}\backslash \mathcal{D}$ and $x\in \mathcal{D}:$ $\Theta(y)>(|\mathcal{P}|-1)\kappa=\Theta(x)$
\end{lemma}
\begin{proof}
Let $y\in \mathcal{P}\backslash \mathcal{D}$. Since $\vert \mathcal{B}_{\mathfrak{T}_y}\vert\geq|\mathcal{P}|-1$ and using Lemma \ref{bep:bep10}, 
\be \nonumber 
\Theta(y)=U(\mathfrak{T}_y)&\geq& \sum_{i}^{|\mathcal{P}|-1} U(v_i,w_i)+U(v,w) \\ \nonumber
&>&(|\mathcal{P}|-1)\kappa=\Theta(x).
\ee
which holds for any in-tree $\mathcal{T}_y$.
\end{proof}

\begin{lemma}
For all $x\in \mathcal{G}\backslash \mathcal{D}$ and for all trees $\mathcal{T}, \Theta(x)\geq (|\mathcal{P}|-1)\kappa$.
\end{lemma}
\begin{proof} 
There are three cases,
\begin{enumerate}
\item For $x\notin \mathcal{E}\cup \mathcal{P}$, i.e. $x$ lies not in an attractor. Then $\vert\mathcal{B}_{\mathfrak{T}_x}\vert\geq |\mathcal{P}|$ so that $\Theta(x)\geq |\mathcal{P}|\kappa$. 
\item For $x\in\mathcal{P}\backslash \mathcal{D}$, the claim follows from Lemma \ref{bep:bep11}.
\item For $x\in \mathcal{E}$. Denote $A_x$ the attractor in which $x$ lies, then there exists a $y\in \mathcal{P}: y\in A_x$. Since $x$ is a preferential successor to $y$ (and vice-versa), $\mathcal{T}_y\backslash (x,y) \cup (y,x)$ extends to an in-tree $\mathcal{T}_x$ in $x$ such that 
\be\nonumber
 U(\mathfrak{T}_x)&=&U(\mathfrak{T}_y\backslash (x,y))+U(y,x)\\\nonumber
&=&U(\mathfrak{T}_y).
\ee 
Since $U(\mathfrak{T}_{y\in \mathcal{P}})\geq(|\mathcal{P}|-1)\kappa$ the claim is proved.
\end{enumerate}
\end{proof}

We can now finish the proof of Theorem 1.  Since both $\Gamma(x)$ and $\Theta(x)$ are constant on ${\cal D}$, we have that
$\Psi$ is constant on $\cal D$, i.e., $\Psi(x)=\Psi(y)$ for $x,y\in {\cal D}$.
On the other hand, the accessibility $U(\mathcal{T}_x)$ is minimal for states in $\mathcal{D}$, and states that lie in attractors spanned by states in $\mathcal{D}$. Finally, the lifetime $\Gamma(x)$ is maximal for states in $\mathcal{P}$.  As $\mathcal{D}\subset \mathcal{P}$ the states in $\mathcal{D}$ have maximal occupation:
 $\Psi(x)>\Psi(y)$ for all $x\in {\cal D}, y \in {\mathcal G}_N\setminus {\cal D}$.  That concludes the proof.

\section{Zero coupling: boundary driven exclusion process}\label{sse}
When $\kappa=0$, the only interaction is that of on-site exclusion. That exclusion process enjoys a matrix representation. In \cite{der2} it was shown that the probability of a configuration $x=(x(1),x(2),\hdots, x(N))$ can be written as 
\be\label{mx}
\rho\left(x\right)=\frac{ \langle W \vert X_1 \hdots X_n \vert V\rangle}{  \langle W \vert (D+E)^N\vert V\rangle }
\ee
where the matrix $X_i$ depends on the occupation $x(i)$ of site $i$ by
\[
X_i= x(i) D + (1-x(i))E
\]
and the matrices $D$ and $E$ satisfy the algebraic rules
\be
&&DE-ED=D+E\nonumber\\
&&\langle W \vert(\alpha E - \gamma D)=\langle W\vert \label{rux}\\\nonumber
&& (\sigma D -\delta E )\vert V \rangle =\vert V\rangle .
\ee
In our model, $\alpha(\beta)= e^{\beta \mu_L}$, $\gamma(\beta)=e^{-\beta \mu_L}$, $\sigma(\beta)=e^{-\beta \mu_R}$ and $\delta(\beta)=e^{\beta \mu_R}$ where we of course insisted on the physical dependence on the environment temperature and we fix  
$\mu_L >0 > \mu_R$ as we take $\beta$ to infinity, which is also the case for Theorem \ref{thm:thm1}.

\begin{theorem} For $N\geq 2$ and with $\mu_L >0 > \mu_R$, 
\be
\rho(x)\asymp 1 \mbox{ iff } x(1)=1, x(N)=0
\ee
\end{theorem}
\begin{proof}
We can immediately put $\gamma,\delta=0$ in \eqref{rux}. The algebraic rules then become 
\be\nonumber\label{algrul}
&&DE-ED=D+E\\
&&\langle W \vert \alpha(\beta) E =\langle W\vert \\\nonumber
&&  D \sigma(\beta) \vert  V \rangle =\vert V\rangle .
\ee
The probability of a configuration $x=(1,x(2),\hdots , x(N-1),0)$ will in the required limit tend to the limit of
\be\label{eq:eq1}
\rho\left(x\right)= \frac{\langle W \vert D \prod_i^p (D^{N_i}E^{M_i}) E \vert V\rangle}{   \langle W \vert (D+E)^N\vert V\rangle      }
\ee
where the product is ordered from left to right and we have arranged the product of matrices in \eqref{mx} in consecutive blocks of product of $D$ and $E$.
First consider the case where $p=1$: 
\be \nonumber \label{weight1}
\langle W \vert D (D^{N_1}E^{M_1}) E \vert V\rangle &=& \langle W\vert \sum_{k=1} a_{k} D^k +\sum_{k=1} b_k E^k + \sum_{i,j\geq 1} c_{i,j} E^iD^j\vert V\rangle\\
&=&\sum_{k=1}\frac{a_k}{\sigma(\beta)^k}+\sum_{k=1}\frac{b_k}{\alpha(\beta)^k}+\sum_{i,j\geq 1}\frac{c_{i,j}}{\alpha(\beta)^i\sigma(\beta)^j}\langle W\vert V\rangle
\ee
where $a_i,b_i$ and $c_{i,j}$ are the coefficients found when expanding $D^{N_1}E^{M_1}$ using the first equation in (\ref{algrul}). In particular $a_1=b_1>0$. Then the dominant contribution to (\ref{weight1}) comes from the $k=1$ term,
\be
&&\langle W \vert D (D^{N_1}E^{M_1}) E \vert V\rangle = \left(a_1\left(\frac{1}{\alpha(\beta)}+\frac{1}{\sigma(\beta)}\right)+O(e^{-2\beta})\right)\langle W\vert V\rangle
\ee
From equation (57) in \cite{der2}, when $\gamma=\delta=0$, it follows that
\be\label{norma}
\frac{\langle W\vert (D+E)^N \vert V\rangle}{\langle W\vert V\rangle}=\frac{\Gamma(N+\frac{1}{\alpha(\beta)} +\frac{1}{\sigma(\beta)})}{\Gamma(\frac{1}{\alpha(\beta)} +\frac{1}{\sigma(\beta)})}
\ee
where $\Gamma(z)$ is the Gamma function which satisfies $\Gamma(z+1)=z\Gamma(z)$. Let $w(\beta)=\frac{1}{\alpha(\beta)}+\frac{1}{\sigma(\beta)}$, then again to significant order,
\be\nonumber
\log \rho\left(x\right)&=&  \log\left(a_1\frac{w(\beta)\Gamma(w(\beta))}{\Gamma(N+w(\beta))}\right)  = \log\left(a_1\frac{\Gamma(w(\beta)+1)}{\Gamma(N+w(\beta))}\right) \\ 
&=& \log\left(a_1\prod\limits_{i=1}^{N-1} \frac{1}{w(\beta)+i}\right)= \log(a_1)-\sum_{i=1}^{N-1}\log\left( w(\beta)+i\right)
\ee
so that $\lim_{\beta\rightarrow \infty} \frac{1}{\beta} \log \rho(x)=0$. Assume this is true for the ordered product of $p$ such factors, that is 
\be\label{kp}
\lim_{\beta\rightarrow \infty}\frac{1}{\beta}\log\rho(x)=\lim_{\beta\rightarrow \infty}\frac{1}{\beta} \log\left(\frac{\langle W \vert D \prod_i^p (D^{N_i}E^{M_i}) E \vert V\rangle}{   \langle W \vert (D+E)^N\vert V\rangle      }\right)=0.
\ee
For $p+1$ factors we then write
\be \label{kp2}\nonumber
&&\langle W \vert D\left( \prod_{i=1}^p D^{N_i}E^{M_i}\right) (D^{N_{p+1}}E^{M_{p+1}}) E \vert V\rangle = \langle W \vert D \prod_{i=1}^p (D^{N_i}E^{M_i}) (D^{N_{p+1}}E^{M_{p+1}+1}) \vert V\rangle \\ \nonumber
&=&\sum_{k=1} a_k\langle W \vert D  \left(\prod_{i=1}^p D^{N_i}E^{M_i}\right) D^k  \vert V\rangle + \sum_{k=1} b_k\langle W \vert D \left(\prod_{i=1}^p D^{N_i}E^{M_i}\right) E^k  \vert V\rangle   \\ \nonumber
&&+ \sum_{\ell,j\geq 1} c_{\ell,j}\langle W \vert D \left(\prod_{i=1}^p D^{N_i}E^{M_i}\right) E^\ell D^j  \vert V\rangle  \\ \nonumber
&=&\sum_{k=1} \frac{a_k}{\sigma(\beta)^k}\langle W \vert D  \left(\prod_{i=1}^p D^{N_i}E^{M_i}\right)  \vert V\rangle + \sum_{k=1} b_k\langle W \vert D \left(\prod_{i=1}^p D^{N_i}E^{M_i'}\right) \vert V\rangle  \\ 
&&+ \sum_{\ell,j\geq 1} \frac{c_{\ell,j}}{\sigma(\beta)^j} \langle W \vert D \left(\prod_{i=1}^p D^{N_i}E^{M_i''}\right)   \vert V\rangle 
\ee
where $M_i'=M_i''=M_i$ for $i\neq p$, $M_p'=M_p+k$ and $M_p''=M_p+\ell$. The first and last term in (\ref{kp2}) are exponentially smaller than the second term due to the extra $\frac{1}{\sigma(\beta)}$ in the denominator. For large $\beta$ one can therefore neglect both the first and last term. The surviving term is exactly the same as in the nominator of (\ref{kp}). Hence (\ref{kp}) is true for any $p>0$.
\newline
 To show that configurations like in Theorem 6.1 are the only configurations such that $\rho(x)\asymp 1$, consider for instance the following one $x=(0,x(2),\hdots , x(N-1),1)$. Its probability is
\be\label{kp3}
\rho(x)=\frac{\langle W \vert E \prod_i^p (D^{N_i}E^{M_i})D \vert V\rangle}{   \langle W \vert (D+E)^N\vert V\rangle      }
\ee
again using the ordered product.
To proceed by induction, we again consider first the case when $p=1$:
\be \label{kp4}
&&\langle W \vert E (D^{N_1}E^{M_1}) D \vert V\rangle =\frac{1}{\sigma(\beta)\alpha(\beta)} \langle W \vert  D^{N_1}E^{M_1} \vert V\rangle
\ee
We insert the expansion $D^{N_1}E^{M_1}=\sum_k a_kD^k + \sum_k b_k E^k + \sum_{\ell,j}c_{\ell,j}E^\ell D^j$ in (\ref{kp4}). This yields to significant order
\be
\left(\frac{a_1}{\sigma(\beta)\alpha(\beta)}\left(\frac{1}{\alpha(\beta)}+\frac{1}{\sigma(\beta)}\right)+O(e^{-2\beta})\right)\langle W\vert V\rangle.
\ee
Let $w(\beta)=\frac{1}{\alpha(\beta)}+\frac{1}{\sigma(\beta)}$, then again to significant order in $\beta$,
\be\nonumber
\log \rho\left(x\right)&=&  \log\left(a_1\frac{w(\beta)}{\sigma(\beta)\alpha(\beta)}\frac{\Gamma(w(\beta))}{\Gamma(N+w(\beta))}\right)  \\
&=& \log\left(a_1\frac{1}{\sigma(\beta)\alpha(\beta)}\frac{\Gamma(w(\beta)+1)}{\Gamma(N+w(\beta))}\right)
\ee
using the normalisation (\ref{norma}) for the first equality above. Hence,
\be
\lim_{\beta\rightarrow \infty }\frac{1}{\beta}\log \rho\left(x\right)&=&\lim_{\beta\rightarrow \infty } \frac{1}{\beta}\log\left(\frac{a_1}{\sigma(\beta)\alpha(\beta)}\frac{\Gamma(w(\beta)+1)}{\Gamma(N+w(\beta))}\right)=  \mu_R-\mu_L <0.
\ee
Assume that this holds for $p$ products, that is 
\be\label{kp5}
\lim_{\beta\rightarrow \infty}\frac{1}{\beta}\log\rho(x)=\lim_{\beta\rightarrow \infty}\frac{1}{\beta} \log\left(\frac{\langle W \vert E \prod_i^p (D^{N_i}E^{M_i}) D \vert V\rangle}{   \langle W \vert (D+E)^N\vert V\rangle      }\right)<0.
\ee
Then for $p+1$ products
\be \label{kp6}\nonumber
&&\langle W \vert E\left( \prod_{i=1}^p D^{N_i}E^{M_i}\right) (D^{N_{p+1}}E^{M_{p+1}}) D \vert V\rangle =\frac{1}{\sigma(\beta)\alpha(\beta)} \langle W \vert  \prod_{i=1}^p (D^{N_i}E^{M_i}) (D^{N_{p+1}}E^{M_{p+1}}) \vert V\rangle \\ \nonumber
&=&\sum_{k=1} \frac{a_k}{\sigma(\beta)^{k+1}\alpha(\beta)}\langle W \vert   \left(\prod_{i=1}^p D^{N_i}E^{M_i}\right)  \vert V\rangle + \sum_{k=1} \frac{b_k}{\sigma(\beta)\alpha(\beta)}\langle W \vert  \left(\prod_{i=1}^p D^{N_i}E^{M_i'}\right) \vert V\rangle  \\ 
&&+ \sum_{\ell,j\geq 1} \frac{c_{\ell,j}}{\sigma(\beta)^{j+1}\alpha(\beta)} \langle W \vert \left(\prod_{i=1}^p D^{N_i}E^{M_i''}\right)   \vert V\rangle 
\ee
All terms that appear in (\ref{kp6}) decay exponentially to zero. The slowest decaying term is clearly the second one so that for large $\beta$ we can neglect both the first and last. The surviving term is exacly the same as in the nominator of (\ref{kp5}). Hence by induction this is true for all $p$. Note that the difference with expression (\ref{kp2}) is the $\frac{1}{\sigma(\beta)\alpha(\beta)}$ factor that makes (\ref{kp6}) to decay exponentially. \newline
It is very similar to show that the other remaning configurations (such as $\{0,x(1),\hdots , x(N-1),0\}$) decay exponentially to zero as temperature goes to zero.
\end{proof} 
\noindent {\bf Acknowledgments}:  We very much thank Karel Neto\v{c}n\'{y} for the many discussions on this topic.  In particular, WOKdG  is grateful for the hospitality at the Institute of Physics, Academy of Sciences in Prague.

\end{document}